\documentclass[twoside]{article}

\usepackage[T2A]{fontenc}
\usepackage[tbtags]{amsmath}
\usepackage{amsthm}
\usepackage{amsfonts,amssymb,mathrsfs}

\theoremstyle{plain}
\newtheorem{theorem}{Theorem}
\newtheorem{lemma}{Lemma}
\newtheorem{corollary}{Corollary}

\theoremstyle{definition}

\newtheorem{remark}{Remark}

\begin{document}

\title{Model of the $p$-Adic Random Walk in a Potential}

\author{A.\,Kh.~Bikulov\thanks{Institute of Chemical Physics, Kosygina Street 4, 117734 Moscow, Russia. E-mail: beecul@mail.ru}\  \ and
 A.\,P.~Zubarev \thanks{Physics Department, Samara State Aerospace University, Moskovskoe shosse 34, 443123, Samara, Russia. E-mail: apzubarev@mail.ru}
\thanks{Physics and Chemistry Department, Samara State University of Railway Transport, Perviy Bezimyaniy pereulok 18, 443066, Samara, Russia.}}
\maketitle

\begin{abstract}
We consider the $p$-adic random walk model in a potential, which
can be viewed as a generalization of $p$-adic random walk models
used for description of protein conformational dynamics. This model is based
on the Kolmogorov--Feller equations for the distribution function
defined on the field of $p$-adic numbers in which the probability
of transitions per unit time depends on  ultrametric distance between
the transition points as well as on  function of potential violating
the spatial homogeneity of a random process. This equation, which
will be called the equation of $p$-adic random walk in a potential,
is equivalent to the equation of $p$-adic random walk with modified
measure and reaction source. With a special choice of a power-law
potential the last equation is shown to have an exact analytic solution.
We find the analytic solution of the Cauchy problem for such equation
with an initial condition, whose support lies in the ring of integer
$p$-adic numbers. We also examine the asymptotic behaviour of the
distribution function for large times. It is shown that in the limit
$t\rightarrow\infty$ the distribution function tends to the equilibrium
solution according to the law, which is bounded from above and below
by power laws with the same exponent. Our principal conclusion is
that the introduction of a potential in the ultrametric model of conformational
dynamics of protein conserves the power-law behaviour of relaxation
curves for large times.
\end{abstract}

\section{Introduction}

An ultrametric space is a metric space~$M$ with
metric $d(x_{1},x_{2})$, $x_{1},\,x_{2}\in M$ satisfying the strong
triangle inequality
\begin{equation}
\forall\,x_{1},\,x_{2}\,x_{3}\in M:\:d(x_{1},x_{2})\le\max\{d(x_{1},x_{3}),\:d(x_{2},x_{3})\}.\label{um}
\end{equation}
A metric $d(x_{1},x_{2})$ satisfying (\ref{um}) is called an ultrametric.
A~classical example of an ultrametric space is the field of $p$-adic
numbers $\mathbb{Q}_{p}$ \cite{Koblitz,Mahler,Sh}. The $p$-adic
numbers were introduced in 1889 by K.~Hesel and were proved widely
useful in algebraic geometry, theory of numbers and representation
theory.

For a long time the ultrametric approach have been found useful in
solving various problems in the field of classification of objects
and information processing of data arrays~\cite{RTV}. In the last
30 years the ultrametric analysis had had a~recent reemergence---it
received a~great impetus from the pioneering works of researches
from the scientific school of Academician V.\,S.~Vladimirov, whose
efforts were later taken up by a~number of researches from different
scientific schools (for an overview, consult, for example,~\cite{ALL}).
This relative new research field is now known as the ``$p$-adic
and ultrametric mathematical physics'' and is blessed by a~number
of books and an immense number of research articles in the area of
$p$-adic analysis, $p$-adic mathematical physics and their applications
to modeling in various areas of physics, biology, computer science,
sociology, physiology, etc.\ (see, for example, \cite{ALL,VVZ} and
the references given therein).

First ultrametric models in physics emerged first in the 1970s in
the theory of spin systems with disorder \cite{RTV,Dayson,Parisi1,Dot}.
It was found that if a~system has a~large number of multiscale ``internal
contradictions'' (frustrations), then it may reach an equilibrium
in hierarchically nested regions of the phase space, the nesting level
increases with decreasing temperature. Here, the ratio of scales of
phase domains satisfies the strong triangle inequality, and thus the
low-temperature spin states are found to be ultrametrically correlated.
Almost immediately after the emergence of the ultrametric model of
spin glass, a~conjecture was made about the ultrametricity of the
space of conformational states of a~protein molecule \cite{Frauen1,Frauen2}.
It should be noted, however, that a~literal transfer of the ultrametric
description of phase states of the spin glass and conformational states
of protein molecule was triggered by a~then-popular analogy between
the low-temperature proteins and glasses~\cite{KG}. As in the case
of glasses, the atomic mobility of polymer globules is accompanied,
in particular, by multiscale constraints generating highly intersected
energy landscapes. In this sense, proteins resemble glasses, some
of their thermodynamical properties are the same (for example, the
low-temperature behaviour of heat capacities). But, as distinct from
disordered globular structures, proteins have a~unique special feature---they
are capable of performing precise operations with objects of atomic
scale, the physical reasons behind it remaining unclear.

A systematic treatment of the conformational dynamics of proteins
based on an ultrametric model is given in a~series of papers by the
authors \cite{ABK_1999,ABKO_2002,ABO_2003,ABO_2004,AB_2008,ABZ_2009,ABZ_2011,ABZ_2013,ABZ_2014}.
In these papers, the conformational dynamics of protein was described
as follows. The conformational state of a~protein was considered
as a~quasi-equilibrium macrostate, which includes many microstates
of a~macromolecule. A~conformational state which is understood in
exactly this sense is also called a~``basin''. Further, there is
a great number of bonds on mechanical degrees of freedom of atoms
in a~macromolecule. As a~result, there are many constraints on the
flexibility of a~macromolecule at various scales of the configuration
space of degrees of freedom of the entire macromolecule. It is assumed
that a~random walk of a~protein over the configuration space is
effected as follows. The most probable are such motions of atoms in
a~macromolecule which take place in small regions of the configuration
space (small-scale rearrangements of atoms). In order that there would
be a~rearrangement of atoms on a~large scale involving a~much larger
region of the configuration space, it is necessary that rearrangements
of smaller scales would result in a~configuration of a~macromolecule
from which a~passage to a~different region of the configuration
space would be enabled by the available bonds. The probability of
such an event is much smaller than the transition probability between
arbitrary small-scale configurations. Hence, the characteristic times
of probability transitions between configurations are much smaller
on more fine scales that those of probability transitions between
configurations on more coarse scales of the configuration space of
a~macromolecule. In accordance with this approach, one may conventionally
single out an increasing family of scales (characteristic sizes of
the regions) of the configuration space of a~macromolecule. Correspondingly,
the characteristic times of a~macromolecule to stay in regions of
the configuration space of various scales increase with increasing
scale. From such a model of a~random walk of a~protein over regions
of the configuration space (conformations or basins) of various scales
one may catch the hierarchy of transitions and the hierarchy of characteristic
times of the transitions. The indexation of such a~hierarchy enables
one to introduce in a~natural way an ultrametric distance between
conformational states. This model features a~number of simplifying
assumptions. Namely, it is assumed that the distance graph is a Kelley
tree and that all conformations are homogeneous in energy. These assumptions
are fairly restrictive, but nevertheless, they enable one to formalize
such a~model by a~$p$-adic random process. In doing so it is assumed
that the probabilities of transitions per unit time through the activation
barriers (that is, transitions between conformational marostates)
obey the Arrhenius relation. Such a~model, which looks simple at
first sight, has enabled one to adequately describe two principal
experiments: the experiment on spectral diffusion in proteins \cite{AB_2008}
and the experiment on binding kinetics of $CO$ by myoglobin \cite{ABZ_2013,ABZ_2014}.
We note that in these two experiments, the fluctuation dynamics of
a~protein is described in a~unified way.

The further development of this model consists in its generalization
by cancelling one or several simplifying assumptions. One of the ways
towards it is to reject the assumption on the homogeneity of conformations
relative to energies. This can be achieved by introducing a~potential
in the model.

The statement of the problem of ultrametric diffusion in a~potential
well is discussed in the literature on $p$-adic mathematical physics
starting from the 2000s. Here, there is a~principal difficulty regarding
the introduction of a~potential in the equation of an ultrametric
random walk. It is known (see, for example, \cite{Gardiner}) that
the right-hand side of the equation for the distribution function
of an arbitrary time-homogeneous Markov process contains three terms.
The first term describes the drift of a~trajectory and is related
with the potential force, the second term describes the diffusion
processes, and the third term describes the jump processes. In the
case when the trajectories of a~random process are everywhere continuous,
the term responsible for the jump processes vanishes and the equation
becomes the Fokker--Planck equation, which contains two terms---the
term responsible for the potential force acting on a~walking particle
and the diffusion term. Since a~$p$-adic random walk is supported
by functions that are constant everywhere except at discontinuities
of the first kind, it follows that the first two terms, including
the term containing the potential, vanish. Containing only the third
term, the equation now becomes the equation, which is conventionally
called the Kolmogorov--Feller equation or the master equation. In
this case, it is natural to introduce a~potential in the form of
some additional structure on an ultrametric space; this additional
structure should be incorporated in the probability of transition
per unit time, which in turn enters the master equation.

One way to introduce potential in the model of a $p$-adic random
walk is to consider the master equation for the probability density
of the form
\begin{equation}
\dfrac{df\left(x,t\right)}{dt}=\intop_{\mathbb{Q}_{p}}\frac{u\left(y\right)f(y,t)-u\left(x\right)f(x,t)}{|x-y|_{p}^{\alpha+1}}d_{p}y,\label{Gen_Eq}
\end{equation}
where $\mathbb{Q}_{p}$ is field of $p$-adic numbers, $d_{p}y$ is
the Haar measure on $\mathbb{Q}_{p}$, $\alpha$ is  positive parameter
interpreted as the inverse temperature, $u\left(x\right)$ is positive
definite function on $\mathbb{Q}_{p}$ (potential function). We
assume that the function $u\left(x\right)$ belongs to the class $W_{\beta}$,
for some $\beta<\alpha$ i.e. satisfies the properties: 1) $|u\left(x\right)|\leq C\left(1+|x|_{p}^{\beta}\right)$;
2) there exists $n$, such that $u\left(x+x'\right)=u\left(x\right)$
for all $|x'|_{p}\leq p^{-n}$.  The Cauchy problem for equation (\ref{Gen_Eq})
is determined in the class of functions $f(x,t)$, which are continuously differentiable in $t$, and as a function of $x \in {\mathbb{Q}_{p}}$  belong to the class class $W_{\gamma}$ for $\gamma<\alpha-\beta$ uniformly with respect to $t$.
In this work, equation (\ref{Gen_Eq}) will be called
the equation of $p$-adic random walk in a~potential $u\left(x\right)$.

At present, methods of analytical solutions of equation (\ref{Gen_Eq})
in potential of arbitrary form $u\left(x\right)$ are unknown. Nevertheless,
it turns out that for a~definite form of a~power-law potential this
equation has a~precise analytic solution.

This article is devoted to the analytical solution of the Cauchy problem
for equation (\ref{Gen_Eq}) with the initial condition on the ring
of integer $p$-adic numbers $\mathbb{Z}_{p}$ for the specific choice
of the potential function $u\left(x\right)$ depending only on $p$-adic
norm $|x|_{p}$ of $x\in\mathbb{Q}_{p}$ as the power law of a special
form. This model can be viewed as a generalization of the models used
in previous works
for description of protein conformational dynamics
\cite{ABK_1999,ABKO_2002,ABO_2003,ABO_2004,AB_2008,ABZ_2009,ABZ_2011,ABZ_2013,ABZ_2014}.
In Section 2 we show that, for such special potentials, the equation
of $p$-adic random walk in a~potential can be reduced to the equation
of $p$-adic random walk with modified measure and reaction sources,
whose support lies in $\mathbb{Z}_{p}$. The last equation will be
the subject of research of the last two sections. In Section 3 we
examine the equation of a~$p$-adic random walk with arbitrary measure
and a~source in~$\mathbb{Z}_{p}$. Using the basis of eigenfunctions
of the Vladimirov operator with modified measure, as was constructed
in our recent paper \cite{BZ_2}, we find an analytic solution of
the Cauchy problem for such an equation with an initial condition
in $\mathbb{Z}_{p}$. In Section 4 we particularize the solution of
this equation in the case when the measure is induced by a~potential
of special form, which was considered in Section 2. In Section 5 we
examine the asymptotic behaviour of the so-obtained solution for the
distribution function for large times. Namely, we establish that the
distribution function tends, as $t\rightarrow\infty$, to the stationary
solution according to a~law that is bounded from above and below
by power laws with the same exponent. Our main conclusion is that
the introduction of a~potential in the ultrametric model of conformational
dynamics of protein conserves the power-law behaviour of the relaxation
curves for large times. In Appendixes A, B, C and D we state and prove
a~number of assertions that are used for solving the principal equation
of the model and in obtaining an asymptotic estimate of the solution
thus obtained for large times.

\section{Particular case of a power-law potential}

We consider the equation of $p$-adic random walk in the potential
(\ref{Gen_Eq}) when the function $u\left(x\right)$ is in the class
$W_{\beta}$ and has the form

\begin{equation}
u\left(x\right)=a\Omega\left(|x|_{p}\right)+b|x|_{p}^{\beta}\left(1-\Omega\left(|x|_{p}\right)\right),\;0\leq\beta<\alpha,\label{u(x)}
\end{equation}
where
\[
\Omega\left(|x|_{p}\right)=\begin{cases}
1, & |x|_{p}\leq1,\\
0, & |x|_{p}>1.
\end{cases}
\]

We denote $L_{l.c.}^{2}\left(\mathbb{Q}_{p},d_{p}x\right)\subset W_{0}$
the class of locally constant functions of functions square-integrable
with respect to $p$-adic Haar measure $d_{p}x$. Also we denote $L_{l.c.}^{2}\left(\mathbb{Q}_{p},u\left(x\right)d_{p}x\right)$
the class of locally constant functions of functions square-integrable
with respect to measure $u\left(x\right)d_{p}x$. Let us state the
Cauchy problem for (\ref{Gen_Eq}) with the potential (\ref{u(x)})
and the initial condition

\begin{equation}
f\left(x,0\right)=f_{0}\left(x\right)\label{ic}
\end{equation}
in the class $L_{l.c.}^{2}\left(\mathbb{Q}_{p},d_{p}x\right)$. The
following result holds.

\begin{lemma} Let constants $a$, $b$, $\alpha$ and $\beta$ satisfy
the conditions
\begin{equation}
a=b\dfrac{1-p^{-1}}{1-p^{-\alpha}},\label{rel_a_b}
\end{equation}
\begin{equation}
\beta=\alpha-1.\label{beta}
\end{equation}
Then the solution of the Cauchy problem for equation
\begin{equation}
\dfrac{df\left(x,t\right)}{dt}=\intop_{\mathbb{Q}_{p}}
\frac{f(y,t)-f(x,t)}{|x-y|_{p}^{\alpha+1}}u\left(y\right)d_{p}y+\lambda\Omega\left(|x|_{p}\right)f(x,t),\label{GE-Omega}
\end{equation}
with the initial condition (\ref{ic}) (here $\lambda=b\Gamma_{p}\left(\alpha\right)\varGamma_{p}\left(-\alpha\right)$
and $\Gamma_{p}(-\alpha)=\dfrac{1-p^{-\alpha-1}}{1-p^{\alpha}}$ is
$p$-adic gamma function) in the class $L_{l.c.}^{2}\left(\mathbb{Q}_{p},u\left(x\right)d_{p}x\right)$
is the solution of the Cauchy problem for equation (\ref{Gen_Eq})
in the class $L_{l.c.}^{2}\left(\mathbb{Q}_{p},d_{p}x\right)$ with
the initial condition (\ref{ic}). \label{lemma1} \end{lemma}

\begin{proof} First of all, note that $L_{l.c.}^{2}\left(\mathbb{Q}_{p},u\left(x\right)d_{p}x\right)\subset L_{l.c.}^{2}\left(\mathbb{Q}_{p},d_{p}x\right)$.
Next, we show that under conditions (\ref{rel_a_b}) and (\ref{beta})
the equation (\ref{Gen_Eq}) can be transformed to the form (\ref{GE-Omega}).
The equation (\ref{Gen_Eq}) can be rewritten in the form of the equation
of $p$-adic random walk with the modified measure $u\left(y\right)d_{p}y$
and reaction source as follows
\begin{equation}
\dfrac{df\left(x,t\right)}{dt}=\intop_{\mathbb{Q}_{p}}\frac{f(y,t)-f(x,t)}{|x-y|_{p}^{\alpha+1}}u\left(y\right)d_{p}y+J\left(x\right)f(x,t).\label{EQ_mod}
\end{equation}
where

\begin{equation}
J\left(x\right)=\varGamma_{p}\left(-\alpha\right)D^{\alpha}u\left(x\right),\label{J(x)}
\end{equation}
$D^{\alpha}$ is the pseudodifferential Vladimirov operator \cite{VVZ}
\[
D^{\alpha}\varphi\left(x\right)=\dfrac{1}{\varGamma_{p}\left(-\alpha\right)}\intop_{\mathbb{Q}_{p}}\frac{\varphi\left(y\right)-\varphi\left(x\right)}{|x-y|_{p}^{\alpha+1}}d_{p}y.
\]
To calculate $D^{\alpha}u\left(x\right)$ we find the $p$-adic Fourier
image $\widetilde{u}\left(k\right)$ of the function $u\left(x\right)$.
Using the relations
\[
\intop_{\mathbb{Q}_{p}}d_{p}x\Omega\left(|x|_{p}\right)\chi\left(kx\right)=\Omega\left(|k|_{p}\right)
\]
and
\[
\intop_{\mathbb{Q}_{p}}d_{p}x|x|_{p}^{\beta}\left(1-\Omega\left(|x|_{p}\right)\right)\chi\left(kx\right)=
\]
\[
=\Gamma_{p}\left(\alpha\beta+1\right)|k|_{p}^{-\beta-1}-\dfrac{1-p^{-1}}{1-p^{-\alpha\beta-1}}\Omega\left(|k|_{p}\right)-\Gamma_{p}\left(\alpha\beta+1\right)|k|_{p}^{-\beta-1}\left(1-\Omega\left(|k|_{p}\right)\right)
\]
we find that
\[
\widetilde{u}\left(k\right)=\intop_{\mathbb{Q}_{p}}d_{p}xu\left(x\right)\chi\left(kx\right)=
\]
\[
=a\Omega\left(|k|_{p}\right)+b\left(-\dfrac{1-p^{-1}}{1-p^{-\beta-1}}\Omega\left(|k|_{p}\right)+\Gamma_{p}\left(\beta+1\right)|k|_{p}^{-\beta-1}\Omega\left(|k|_{p}\right)\right).
\]
We have
\[
D^{\alpha}u\left(x\right)=\intop_{\mathbb{Q}_{p}}d_{p}k|k|_{p}^{\alpha}\widetilde{u}\left(k\right)\chi\left(-kx\right),
\]
and hence
\[
D^{\alpha}u\left(x\right)=\left(a-b\dfrac{1-p^{-1}}{1-p^{-\beta-1}}\right)D^{\alpha}\Omega\left(|x|_{p}\right)+b\Gamma_{p}\left(\alpha\right)D^{-\beta+\alpha-1}\Omega\left(|x|_{p}\right)
\]
In view of conditions (\ref{rel_a_b}) and (\ref{beta}), we find
that
\[
D^{\alpha}u\left(x\right)=b\Gamma_{p}\left(\alpha\right)\Omega\left(|x|_{p}\right).
\]
As a result,
\[
J\left(x\right)=\intop_{\mathbb{Q}_{p}}\frac{u\left(y\right)-u\left(x\right)}{|x-y|_{p}^{\alpha+1}}dy=\varGamma_{p}\left(-\alpha\right)D^{\alpha}u\left(x\right)=b\Gamma_{p}\left(\alpha\right)\varGamma_{p}\left(-\alpha\right)\Omega\left(|x|_{p}\right)=\lambda\Omega\left(|x|_{p}\right)
\]
and so (\ref{EQ_mod}) assumes the form (\ref{GE-Omega}), which completes
the proof of Lemma~\ref{lemma1}.\end{proof}

\begin{remark} Note that in the framework of assumptions (\ref{rel_a_b})
and (\ref{beta}) of Lemma~\ref{lemma1} we have $\lambda>0$ for
$\alpha>1$, $\lambda<0$ for $\alpha<1$, and $\lambda=0$
for $\alpha=1$. Moreover, for $\alpha<1$ we have $\beta<0$,
hence we will study only the case $\alpha>1$.\label{rem1}\end{remark}

\begin{remark} For $\alpha>2$ the equation has the normalized
to unity stationary solution
\begin{equation}
f^{(st)}\left(x\right)=\left(\dfrac{1-p^{-\alpha}}{1-p^{-1}}+\dfrac{1-p^{-1}}{p^{\alpha-2}-1}\right)^{-1}\left(\dfrac{1-p^{-\alpha}}{1-p^{-1}}\Omega\left(|x|_{p}\right)+|x|_{p}^{-\alpha+1}\left(1-\Omega\left(|x|_{p}\right)\right)\right).\label{f^(st)}
\end{equation}
Later we will show that the equilibrium solution of equation (\ref{Gen_Eq})
\[
f^{(eq)}\left(x\right)=\lim_{t\rightarrow\infty}f\left(x,t\right)
\]
is
\[
f^{(eq)}\left(x\right)=\begin{cases}
f^{(st)}\left(x\right), & \alpha>2,\\
0, & 1\leq\alpha\leq 2.
\end{cases}
\]
\label{rem2}\end{remark}

\[
f^{(eq)}\left(x\right)=\begin{cases}
f^{(st)}\left(x\right), & \alpha>2,\\
0, & 1\leq\alpha\leq 2.
\end{cases}
\]

Solution of the Cauchy problem for the equation (\ref{GE-Omega})
in the class $L_{l.c.}^{2}\left(\mathbb{Q}_{p},u\left(x\right)d_{p}x\right)$
can be found analytically. Below we shall find and explore an analytical
solution of this problem for the homogeneous initial condition in
$\mathbb{Z}_{p}$.

\section{The general solution of the equation of $p$-adic random walk
with general modified measure $u\left(x\right)d_{p}x$ and reaction
source in~$\mathbb{Z}_{p}$ }

In this section we solve the Cauchy problem of equation

\begin{equation}
\dfrac{df\left(x,t\right)}{dt}=\intop_{\mathbb{Q}_{p}}\frac{f(y,t)-f(x,t)}{|x-y|_{p}^{\alpha+1}}u\left(y\right)d_{p}y+\lambda\Omega\left(|x|_{p}\right)f(x,t)\label{GE_m(x)_Omega}
\end{equation}
with the initial condition
\begin{equation}
f\left(x,0\right)=\varOmega\left(|x|_{p}\right)\label{ic_Vl}
\end{equation}
and an arbitrary function $u\left(x\right)$. In finding the solution
we shall use a~number of auxiliary results, which will be formulated
and proved in Appendix~A.

The solution of the Cauchy problem (\ref{GE_m(x)_Omega})--(\ref{ic_Vl})
will be searched in the subclass of locally constant functions of
the space $L_{norm}^{2}\left(\mathbb{Q}_{p},\:u\left(x\right)d_{p}x\right)$
(see Appendix~A). Passing to the Laplace images, we obtain
\begin{equation}
s\widetilde{f}\left(x,s\right)=\varOmega\left(|x|_{p}\right)+D_{u(x)}^{\alpha}\widetilde{f}\left(x,s\right)+\lambda\varOmega\left(|x|_{p}\right)\widetilde{f}\left(x,s\right).\label{Eq_Vl_Lap}
\end{equation}

Employing Theorem~\ref{th1} from Appendix~A, we expand the function
$\widetilde{f}\left(x,s\right)$ in basis (\ref{basis_f}),
\begin{equation}
\widetilde{f}\left(x,s\right)=\sum_{i=-\infty}^{\infty}\widetilde{f}_{i}\left(s\right)\phi_{i}\left(x\right).\label{decomp_basis}
\end{equation}

An application of Theorems \ref{th2} and \ref{th3} and Corollary
\ref{cor1} from Appendix~A shows that
\[
s\widetilde{f}_{k}\left(s\right)=V_{0}V_{k-1}^{-\tfrac{1}{2}}\left(1-\dfrac{V_{k-1}}{V_{k}}\right)^{\tfrac{1}{2}}-\left(1-p^{-\left(\alpha+1\right)}\right)\left(\sum_{i=k}^{\infty}p^{-i\left(\alpha+1\right)}V_{i}\right)\widetilde{f}_{k}\left(s\right)+
\]
\[
+\lambda\left(V_{0}V_{k-1}^{-1}\left(1-\dfrac{V_{k-1}}{V_{k}}\right)^{-1}\left(1-3\dfrac{V_{k-1}}{V_{k}}+2\dfrac{V_{k-1}^{2}}{V_{k}^{2}}\right)+\sum_{i=k}^{\infty}V_{0}V_{i-1}^{-1}\left(1-\dfrac{V_{i-1}}{V_{i}}\right)\right)\widetilde{f}_{k}\left(s\right)+
\]
\[
-2\lambda V_{0}V_{k-1}^{-1}\left(1-\dfrac{V_{k-1}}{V_{k}}\right)\widetilde{f}_{k}\left(s\right)+
\]
\begin{equation}
+\lambda V_{0}V_{k-1}^{-\tfrac{1}{2}}\left(1-\dfrac{V_{k-1}}{V_{k}}\right)^{\tfrac{1}{2}}\sum_{i=1}^{\infty}f_{i}\left(s\right)V_{i-1}^{-\tfrac{1}{2}}\left(1-\dfrac{V_{i-1}}{V_{i}}\right)^{\tfrac{1}{2}}\label{k_ge_1}
\end{equation}
for $k\geq1$ and
\begin{equation}
s\widetilde{f}_{k}\left(s\right)=-\left(1-p^{-\left(\alpha+1\right)}\right)\left(\sum_{i=k}^{\infty}p^{-i\left(\alpha+1\right)}V_{i}\right)\widetilde{f}_{k}\left(s\right)+\lambda\sum_{i=k+1}^{\infty}\sum_{i=k+1}^{\infty}V_{0}V_{i-1}^{-1}\left(1-\dfrac{V_{i-1}}{V_{i}}\right)\widetilde{f}_{k}\left(s\right)\label{k_l_1}
\end{equation}
for $k<1$.

Equation (\ref{k_ge_1}) can be rewritten as
\begin{equation}
\widetilde{f}_{k}\left(s\right)=g_{k}\left(s\right)+\lambda g_{k}\left(s\right)\sum_{i=1}^{\infty}\widetilde{f}_{i}\left(s\right)V_{i-1}^{-\tfrac{1}{2}}\left(1-\dfrac{V_{i-1}}{V_{i}}\right)^{\tfrac{1}{2}}\label{f(s)}
\end{equation}
where
\[
g_{k}\left(s\right)=V_{0}V_{k-1}^{-\tfrac{1}{2}}\left(1-\dfrac{V_{k-1}}{V_{k}}\right)^{\tfrac{1}{2}}\left[s+\left(1-p^{-\left(\alpha+1\right)}\right)\left(\sum_{i=k}^{\infty}p^{-i\left(\alpha+1\right)}V_{i}\right)-\right.
\]
\[
-\lambda V_{0}V_{k-1}^{-1}\left(1-\dfrac{V_{k-1}}{V_{k}}\right)^{-1}\left(1-3\dfrac{V_{k-1}}{V_{k}}+2\dfrac{V_{k-1}^{2}}{V_{k}^{2}}\right)-
\]
\begin{equation}
\left.-\lambda\sum_{i=k}^{\infty}V_{0}V_{i-1}^{-1}\left(1-\dfrac{V_{i-1}}{V_{i}}\right)+2\lambda V_{0}V_{k-1}^{-1}\left(1-\dfrac{V_{k-1}}{V_{k}}\right)\right]^{-1}.\label{g_k(s)}
\end{equation}
Setting
\[
F\left(s\right)=\sum_{i=1}^{\infty}\widetilde{f}_{i}\left(s\right)V_{i-1}^{-\tfrac{1}{2}}\left(1-\dfrac{V_{i-1}}{V_{i}}\right)^{\tfrac{1}{2}},
\]
\[
G\left(s\right)=\sum_{i=1}^{\infty}g_{i}\left(s\right)V_{i-1}^{-\tfrac{1}{2}}\left(1-\dfrac{V_{i-1}}{V_{i}}\right)^{\tfrac{1}{2}},
\]
we write equation (\ref{f(s)}) in the form
\begin{equation}
\widetilde{f}_{k}\left(s\right)=g_{k}\left(s\right)\left(1+\lambda F\left(s\right)\right).\label{f(s)_F}
\end{equation}
Multiplying (\ref{f(s)_F}) by $V_{i-1}^{-\tfrac{1}{2}}\left(1-\dfrac{V_{i-1}}{V_{i}}\right)^{\tfrac{1}{2}}$
and summing in~$k$, we find that
\[
F\left(s\right)=G\left(s\right)+\lambda G\left(s\right)F\left(s\right),
\]
and hence,
\begin{equation}
F\left(s\right)=\dfrac{G\left(s\right)}{1-\lambda G\left(s\right)}.\label{F(s)}
\end{equation}
Substituting (\ref{F(s)}) into (\ref{f(s)_F}), this gives
\begin{equation}
\widetilde{f}_{k}\left(s\right)=\dfrac{g_{k}\left(s\right)}{1-\lambda G\left(s\right)}.\label{sol_f(s)}
\end{equation}
From equation (\ref{k_l_1}) we find that
\begin{equation}
\widetilde{f}_{k}\left(s\right)=0,\:k<1.\label{til_f_l_1}
\end{equation}
In view of (\ref{sol_f(s)}) and (\ref{til_f_l_1}), function (\ref{decomp_basis})
assumes the form
\begin{equation}
\widetilde{f}\left(x,s\right)=\sum_{i=1}^{\infty}\dfrac{g_{k}\left(s\right)}{1-\lambda G\left(s\right)}\phi_{k}\left(x\right)\label{gen_solution}
\end{equation}
Expression (\ref{gen_solution}) is the general solution of the equation
of $p$-adic random walk with source in~$\mathbb{Z}_{p}$ with the
initial distribution in Laplace images.

\section{Solution in the case of a modified measure induced by a~power-law
potential}

In this section we particularize the solution (\ref{gen_solution})
with the following function $u\left(x\right)$:
\begin{equation}
u\left(x\right)=b\left(\dfrac{1-p^{-1}}{1-p^{-\alpha}}\Omega\left(|x|_{p}\right)+|x|_{p}^{\alpha-1}\left(1-\Omega\left(|x|_{p}\right)\right)\right).\label{u(x)_our}
\end{equation}
Calculating the volumes of balls (\ref{V_i_gen}) with respect to
the measure $u\left(x\right)d_{p}x$, this gives
\begin{equation}
V_{i}=b\left(1-p^{-1}\right)\dfrac{p^{\alpha\left(i+1\right)}}{p^{\alpha}-1}.\label{V_i}
\end{equation}
Here, the basis functions (\ref{basis_f}) assume the form
\begin{equation}
\phi_{i}\left(x\right)=\left(b\left(1-p^{-1}\right)p^{\alpha\left(i-1\right)}\right)^{-\tfrac{1}{2}}\left(\varOmega\left(|x|_{p}p^{-i+1}\right)-p^{-\alpha}\varOmega\left(|x|_{p}p^{-i}\right)\right).\label{phi_i_our}
\end{equation}
Calculating functions (\ref{g_k(s)}), we obtain
\begin{equation}
g_{k}\left(s\right)=\left(bp^{\alpha}\left(1-p^{-1}\right)\right)^{\tfrac{1}{2}}\dfrac{p^{-\alpha k/2}}{s-bp^{-k}p^{\alpha}\Gamma_{p}(-\alpha)}.\label{sol_g(s)}
\end{equation}
The function $G\left(s\right)$ reads as
\begin{equation}
G\left(s\right)=\sum_{i=1}^{\infty}\dfrac{\left(p^{\alpha}-1\right)p^{-\alpha i}}{s-bp^{-i}p^{\alpha}\Gamma_{p}(-\alpha)}.\label{sol_G(s)}
\end{equation}
Substituting (\ref{sol_g(s)}) and (\ref{sol_G(s)}) into (\ref{sol_f(s)}),
we find that
\begin{equation}
\widetilde{f}_{k}\left(s\right)=\left(bp^{\alpha}\left(1-p^{-1}\right)\right)^{\tfrac{1}{2}}\dfrac{p^{-\alpha k/2}}{s-bp^{-k}p^{\alpha}\Gamma_{p}(-\alpha)}\dfrac{1}{1-\lambda\sum_{i=1}^{\infty}\dfrac{\left(p^{\alpha}-1\right)p^{-\alpha i}}{s-bp^{-i}p^{\alpha}\Gamma_{p}(-\alpha)}}.\label{final_f(s)}
\end{equation}
Next, in view of (\ref{phi_i_our}), we have
\[
\widetilde{f}\left(x,s\right)=\sum_{k=1}^{\infty}\dfrac{p^{-\alpha k}p^{\alpha}\left(\varOmega\left(|x|_{p}p^{-k+1}\right)-p^{-\alpha}\varOmega\left(|x|_{p}p^{-k}\right)\right)}{\left(s-bp^{-k}p^{\alpha}\Gamma_{p}(-\alpha)\right)\left(1-\lambda\sum_{i=1}^{\infty}\dfrac{\left(p^{\alpha}-1\right)p^{-\alpha i}}{s-bp^{-i}p^{\alpha}\Gamma_{p}(-\alpha)}\right)}.
\]
Here, the Laplace image of the probability of finding the trajectory
of a~random process in the support of the initial distribution is
as follows:
\[
\intop_{\mathbb{Z}_{p}}\widetilde{f}\left(x,s\right)d_{p}x=\sum_{k=1}^{\infty}\dfrac{p^{-\alpha k}\left(p^{\alpha}-1\right)}{\left(s-bp^{-k}p^{\alpha}\Gamma_{p}(-\alpha)\right)\left(1-\lambda\sum_{i=1}^{\infty}\dfrac{\left(p^{\alpha}-1\right)p^{-\alpha i}}{s-bp^{-i}p^{\alpha}\Gamma_{p}(-\alpha)}\right)}=
\]
\begin{equation}
=\dfrac{G\left(s\right)}{1-\lambda G\left(s\right)}=F\left(s\right).\label{SP}
\end{equation}

\section{Asymptotic behaviour of the distribution function as $t\rightarrow\infty$}

We first examine the asymptotic behaviour as $t\rightarrow\infty$
of the probability of finding the trajectory in the support of the
initial distribution, which we denote by $S\left(t\right)$ and which
is the Laplace preimage of function (\ref{SP}). To this aim we need
to find the poles and residues of the function $G\left(s\right)$
(see~(\ref{sol_G(s)})). Clearly, the function $G\left(s\right)$
is positive at zero, and besides,
\[
G\left(0\right)=-\dfrac{\left(p^{\alpha}-1\right)}{bp^{\alpha}\Gamma_{p}(-\alpha)}\sum_{i=1}^{\infty}p^{-\left(\alpha-1\right)i}=\dfrac{1}{\lambda}.
\]
The function $G\left(s\right)$ vanishes at infinity
\[
\lim_{s\rightarrow\infty}G\left(s\right)=0.
\]
The derivative of $G\left(s\right)$ is negative
\[
\dfrac{d}{ds}G\left(s\right)=-\sum_{i=1}^{\infty}\dfrac{\left(p^{\alpha}-1\right)p^{-\alpha i}}{\left(s-bp^{-i}p^{\alpha}\Gamma_{p}(-\alpha)\right)^{2}} <0,
\]
and hence the function $G\left(s\right)$ is decreasing. The function
$G\left(s\right)$ has simple poles at the points $s_{k}=bp^{-k}p^{\alpha}\Gamma_{p}(-\alpha)<0$,
which correspond to the zeros of the function $F\left(s\right)$.
To the left from the poles $s_{k}=bp^{-k}p^{\alpha}\Gamma_{p}(-\alpha)<0$
the function $G\left(s\right)$ is negative, and to the right, is
positive. The points of intersection of the function $G\left(s\right)$
with the line $G=\dfrac{1}{\lambda}$ correspond to the poles of the
function $F\left(s\right)$. Besides, the function $F\left(s\right)$
has an evident pole at the point $s=0$, at which the function $G\left(s\right)$
assumes the value $\dfrac{1}{\lambda}$. Clearly, all the poles of
the function $F\left(s\right)$ are negative, except for the pole
$s=0$. The function $G\left(s\right)$ vanishes at the infinity,
and hence so is the function $F\left(s\right)$:
\[
\lim_{s\rightarrow\infty}F\left(s\right)=0.
\]

Let $s=-\lambda_{i}$, $i\in\mathbb{Z}_{+}$, be the poles of the
function $F\left(s\right)$. Note that the function $F(s)$ is not
meromorphic (the point $s=0$ is a~point of condensation of simple
poles). Nevertheless, one may show (see~\cite{ABZ_2009} for details)
that the function $F(s)$ can be expanded in simple poles,
\[
F(s)=\dfrac{b_{0}}{s}+\sum_{k=1}^{\infty}\dfrac{b_{k}}{s+\lambda_{k}}.
\]
Clearly, $\lambda_{i}$ satisfy the condition
\[
-bp^{-\left(i+1\right)}p^{\alpha}\Gamma_{p}(-\alpha)\leq\lambda_{i}\leq-bp^{-i}p^{\alpha}\Gamma_{p}(-\alpha).
\]
We write $\lambda_{i}$ as
\[
\lambda_{i}=-bp^{-i-1}p^{\alpha}\Gamma_{p}(-\alpha)\left(1+\delta_{i}\right),
\]
where $\delta_{i}$ is a small positive error satisfying the condition
\begin{equation}
0<\delta_{i}<p-1.\label{delta_i}
\end{equation}
From the equation for the poles,
\[
G\left(-\lambda_{k}\right)=\dfrac{1}{\lambda},
\]
it follows that
\begin{equation}
\sum_{i=1}^{\infty}\dfrac{p^{-\alpha i}}{p^{-k-1}+p^{-k-1}\delta_{k}-p^{-i}}=-\dfrac{1}{p^{\alpha-1}-1}.\label{delta_k}
\end{equation}
Equation (\ref{delta_k}) enables one to obtain an asymptotic estimate
for $\delta_{k}$ as $k\rightarrow\infty$ (see Appendix~B, formula~(\ref{delta_gen})).

The residues at the poles $s=-\lambda_{i}$ of the function $F\left(s\right)$
are as follows:
\[
b_{k}=\lim_{s\rightarrow-\lambda_{k}}\left(s+\lambda_{k}\right)F\left(s\right)=-\dfrac{\lambda^{-2}}{G^{\prime}\left(-\lambda_{k}\right)}=
\]
\begin{equation}
=\dfrac{p^{\alpha}-1}{\left(p^{\alpha-1}-1\right)^{2}}\left[\sum_{i=1}^{\infty}\dfrac{p^{-\alpha i}}{\left(p^{-k-1}+p^{-k-1}\delta_{k}-p^{-i}\right)^{2}}\right]^{-1}.\label{b_k}
\end{equation}
Equation (\ref{b_k}) in view of asymptotic estimate for $\delta_{k}$
enables one to obtain an asymptotic estimate for $b_{k}$ as $k\rightarrow\infty$
(see Appendix~C, formula (\ref{b_gen})). This estimate can be written
as
\[
C_{\min}\left(\alpha\right)p^{-\left|\alpha-2\right|k}\left(1+\delta_{\alpha,2}\left(k^{-2}-1\right)\right)\left(1+o\left(1\right)\right)<
\]
\begin{equation}
<b_{k}<C_{\max}\left(\alpha\right)p^{-\left|\alpha-2\right|k}\left(1+\delta_{\alpha,2}\left(k^{-2}-1\right)\right)\left(1+o\left(1\right)\right),\label{b_gen-1}
\end{equation}
where $o\left(1\right)$ is an infinitesimal quantity as $k\rightarrow\infty$
and
\[
C_{\min}\left(\alpha\right)=\begin{cases}
\dfrac{\left(p^{\alpha}-1\right)\left(1-p^{-1}\right)^{2}\left(1-p^{-\alpha+2}\right)^{2}}{p^{2-\alpha}\left(p^{\alpha-1}-1\right)^{2}}, & \alpha>2,\\
\dfrac{\left(p^{2}-1\right)}{\left(p-1\right)^{2}}, & \alpha=2,\\
\dfrac{\left(p^{\alpha}-1\right)\delta_{\max}^{2}}{p^{2-\alpha}\left(p^{\alpha-1}-1\right)^{2}}, & 1<\alpha<2,
\end{cases}
\]
\[
C_{\max}\left(\alpha\right)=\left(1+o\left(1\right)\right)\begin{cases}
\dfrac{\left(p^{\alpha}-1\right)\left(1-p^{-\alpha+2}\right)^{2}}{p^{2-\alpha}\left(p^{\alpha-1}-1\right)^{2}}, & \alpha>2,\\
\dfrac{\left(1-p^{-2}\right)\left(1-p^{-1}\right)^{2}}{\left(p-1\right)^{2}}, & \alpha=2,\\
\dfrac{\left(p^{\alpha}-1\right)\delta_{\min}^{2}}{p^{2-\alpha}\left(p^{\alpha-1}-1\right)^{2}}, & 1<\alpha<2.
\end{cases}
\]
The function $S\left(t\right)$ is as follows:
\begin{equation}
S(t)=b_{0}+\sum\limits _{k=1}^{\infty}b_{k}\exp\left(-\lambda_{k}t\right).\label{S(t)}
\end{equation}
Using estimates for $b_{k}$ and for $\lambda_{k}$ one can estimate
$S(t)$ as
\[
b_{0}+C_{\min}\left(\alpha\right)\sum_{k}p^{-\left|\alpha-2\right|k}\left(1+\delta_{\alpha,2}\left(k^{-2}-1\right)\right)\exp\left(-wp^{-k}p^{-1}t\right)\left(1+o\left(1\right)\right)<
\]
\[
<S(t)<b_{0}+C_{\max}\left(\alpha\right)\sum_{k}p^{-\left|\alpha-2\right|k}\left(1+\delta_{\alpha,2}\left(k^{-2}-1\right)\right)\exp\left(-wp^{-k}p^{-1}t\right)\left(1+o\left(1\right)\right),
\]
where
\[
w=-bp^{\alpha}\Gamma_{p}(-\alpha).
\]
An application of Lemma \ref{lemma4} from Appendix~D gives, as $t\rightarrow\infty$,
\[
b_{0}+A\left(\alpha\right)t^{-\left|\alpha-2\right|}\left(1+\delta_{\alpha,2}\left(\left(\dfrac{\ln p}{\ln t}\right)^{2}-1\right)\right)\left(1+o\left(1\right)\right)<
\]
\begin{equation}
<S(t)<b_{0}+B\left(\alpha\right)t^{-\left|\alpha-2\right|}\left(1+\delta_{\alpha,2}\left(\left(\dfrac{\ln p}{\ln t}\right)^{2}-1\right)\right)\left(1+o\left(1\right)\right)\label{main_res}
\end{equation}
where $A\left(\alpha\right)$ and $B\left(\alpha\right)$ are defined
by
\[
A\left(\alpha\right)=C_{\min}\left(\alpha\right)p^{-\left|\alpha-2\right|}\left(\ln p\right)^{-1}\left(wp^{-1}\right)^{-\left|\alpha-2\right|}\Gamma\left(\left|\alpha-2\right|\right),
\]
\[
B\left(\alpha\right)=C_{\max}\left(\alpha\right)p^{-\left|\alpha-2\right|}\left(\ln p\right)^{-1}\left(wp^{-1}\right)^{-\left|\alpha-2\right|}\Gamma\left(\left|\alpha-2\right|\right)
\]
and $o\left(1\right)$ is an infinitesimal quantity as $t\rightarrow\infty$.

We note that
\[
b_{0}=\begin{cases}
\dfrac{\left(p^{\alpha}-1\right)\left(p^{\alpha-2}-1\right)}{\left(p^{\alpha-1}-1\right)^{2}} & \alpha>2,\\
0, & \alpha\leq2,
\end{cases}
\]
is nothing else but the equilibrium solution at the point $x=0$,
i.e.

\[
\lim_{t\rightarrow\infty}f\left(0,t\right)=\lim_{t\rightarrow\infty}S\left(t\right)=b_{0},
\]
which for $\alpha>2$ concides with the normalized stationary solution
$f^{(st)}\left(x\right)$ at the point $x=0$
\[
\lim_{t\rightarrow\infty}S\left(t\right)=b_{0}.
\]

Let us discuss the asymptotic behaviour of the distribution function
$f\left(x,t\right)$ as $t\rightarrow\infty$ in the domain $|x|_{p}>1$.
For the Laplace image of $f\left(x,t\right)$ we have
\begin{equation}
\widetilde{f}\left(x,s\right)=\dfrac{K\left(s,|x|_{p}\right)+G\left(s\right)}{1-\lambda G\left(s\right)},\label{f^til(x)}
\end{equation}
where the function $K\left(s,|x|_{p}\right)$ is defined as
\[
K\left(s,|x|_{p}\right)=-\dfrac{p^{\alpha}|x|_{p}^{-\alpha}}{\left(s-b|x|_{p}^{-1}p^{\alpha}\Gamma_{p}(-\alpha)\right)}-\sum_{i=1}^{\ln|x|_{p}/\ln p}\dfrac{p^{-\alpha i}\left(p^{\alpha}-1\right)}{\left(s-bp^{-i}p^{\alpha}\Gamma_{p}(-\alpha)\right)},
\]
and the function $G\left(s\right)$ is defined by formula (\ref{sol_G(s)}).
Clearly, the function (\ref{f^til(x)}) has a~pole at the point $s=0$,
at which the residue $b_{0}\left(x\right)$ is the equilibrium solution
of equation (\ref{GE-Omega}), which for $\alpha>2$ concides with
the normalized stationary solution $f^{(st)}\left(x\right)$. From
(\ref{f^til(x)}) it follows that the function $\widetilde{f}\left(x,s\right)$
also has poles $\lambda_{k}$, which coincide with the poles of the
function $F\left(s\right)$. The residues $b_{k}\left(x\right)$ at
these poles are as follows:
\begin{equation}
b_{k}\left(|x|_{p}\right)=\lim_{s\rightarrow-\lambda_{k}}\left(s+\lambda_{k}\right)\dfrac{K\left(s,|x|_{p}\right)+G\left(s\right)}{\left(1-\lambda G\left(s\right)\right)}.\label{b_k(x)}
\end{equation}
Formula (\ref{b_k(x)}) can be rewritten as
\[
b_{k}\left(|x|_{p}\right)=\left(h_{k}\left(|x|_{p}\right)+1\right)b_{k},
\]
where
\[
h_{k}\left(|x|_{p}\right)=\left(p^{\alpha}-1\right)\left(\dfrac{p^{\alpha}|x|_{p}^{-\alpha}}{\left(p^{-k-1}+p^{-k-1}\delta_{k}-|x|_{p}^{-1}\right)}+\sum_{i=1}^{\ln|x|_{p}/\ln p}\dfrac{p^{-\alpha i}\left(p^{\alpha}-1\right)}{\left(p^{-k-1}+p^{-k-1}\delta_{k}-p^{-i}\right)}\right)
\]
and $b_{k}$ is the residue at the pole $s=-\lambda_{k}$ of the function
$F\left(s\right)$. For any fixed $|x|_{p}$, the function $h_{k}\left(|x|_{p}\right)$
has finite nonzero limit as $k\rightarrow\infty$. Hence, as $k\rightarrow\infty$
the behaviour of the residues $b_{k}\left(|x|_{p}\right)$ and $b_{k}$
are equal up to a~factor depending on $|x|_{p}$. For this reason,
if $t\rightarrow\infty$ and for any $x\in\mathbb{Q}_{p}$ the function
$f\left(x,t\right)$ tends to the equilibrium solution in accordance
with the same law as the function $S\left(t\right)$.

\section{Conclusion}

In the present paper we consider the equation of  $p$-adic random
walk in a potential, which has the form (\ref{Gen_Eq}). It is shown
that in the case when the function of potential has a power form (\ref{u(x)})
and in the particular case when the parameters are related by (\ref{rel_a_b})
and (\ref{beta}), the equation (\ref{Gen_Eq}) can be solved analytically.
We found the solution to the Cauchy problem of this equation with
the initial condition in $\mathbb{Z}_{p}$ and examined its asymptotic
behavior at large times. Moreover, we prove that, as $t\rightarrow\infty$,
the distribution function tends to the stationary solution according
to a~law bounded from above and below by power laws with the same
exponent. The analysis of numerical solutions of the $p$-adic random
walk in a~potential \cite{Zubarev_PC}, which we not present here,
shows that, for other values of the parameter~$\beta$ near the value
$\beta=\alpha-1$, the power-law behaviour of the distribution function
is conserved for large times.

%The model of a $p$-adic random walk in a potential can be viewed
%as a possible generalization of the models of $p$-adic random walk,
%used for modeling protein conformational dynamics.
The results obtained
in this work provide a~basis for the belief that the introduction
of a~power-law potential in the $p$-adic model of the conformational
dynamics of protein conserves the power-law behaviour of the relaxation
curves for large times. This fact may serve as a~justification of
the conclusion that in a~model that is capable of describing experiments
in spectral diffusion in proteins and experiments in binding kinetics
of $CO$ by myoglobin (in which the relaxation has power-law behaviour),
the introduction of a~power-law potential may not be necessary. We
admit, however, that in more precise experiments on the conformational
dynamics of protein the potential may somehow exhibit itself. It is
quite possible that the potential may play a~certain role in models
of prototypes of molecular nanomachines that are based on processes
of $p$-adic random walk. One of such models, which was considered
in~\cite{ABZ_2011}, was based on a~process of $p$-adic random
walk without potential, but on the compact set $B_{r}\subset\mathbb{Q}_{p}$.
Such a~model can be physically interpreted as a~model on~$\mathbb{Q}_{p}$
in a~potential $u(x)$ of a~special form, which is constant for
$x\in B_{r}$ and is infinite for $x\notin B_{r}$.

\vspace{5mm}
{\bf Acknowledgements}

The author is deeply indebted to Prof.\ I.\,V.~Volovich (Steklov Mathematical Institute,
Russian Academy of Sciences) for their  careful reading of the manuscript, discussion of the results, and a~number of useful comments.
The author is also grateful to Prof.\ A.\,R.~Alimov (Faculty of Mechanics and Mathematics,
Moscow State University) for his assistance in preparing the manuscript  and a~number of helpful comments.

\smallskip

The work was partially supported by the Ministry of Education and Science of of the Russian Federation under the Competitiveness Enhancement Program of SSAU for 2013--2020.

\vspace{5mm}

\section*{Appendix~A. Some auxiliary results}

Let $u\left(x\right)$ is positive $d_{p}x$-locally integrable function
on $\mathbb{Q}_{p}$. We denote by $L_{norm}^{2}\left(\mathbb{Q}_{p},\:u\left(x\right)d_{p}x\right)$
the class of functions on~$\mathbb{Q}_{p}$ of the form $f\left(x\right)=f\left(|x|_{p}\right)$
with finite integral
\[
\intop_{\mathbb{Q}_{p}}d_{p}yu\left(y\right)\left|f\left(y\right)\right|^{2}.
\]

\begin{theorem}The family of functions
\begin{equation}
\phi_{i}\left(x\right)=\left(V_{i-1}\left(1-\dfrac{V_{i-1}}{V_{i}}\right)\right)^{-\tfrac{1}{2}}\left(\varOmega\left(|x|_{p}p^{-i+1}\right)-\dfrac{V_{i-1}}{V_{i}}\varOmega\left(|x|_{p}p^{-i}\right)\right),\label{basis_f}
\end{equation}
where $i\in\mathbb{Z}$ and
\begin{equation}
V_{i}=\intop_{\mathbb{Q}_{p}}m\left(y\right)d_{p}y\varOmega\left(|y|_{p}p^{-i}\right),\label{V_i_gen}
\end{equation}
forms an orthonormal basis for $L_{norm}^{2}\left(\mathbb{Q}_{p},\:u\left(x\right)d_{p}x\right)$.\label{th1}\end{theorem}

\begin{proof} That the system of functions (\ref{basis_f}) is orthonormal
\[
\intop_{\mathbb{Q}_{p}}\phi_{i}\left(x\right)\phi_{j}\left(x\right)m\left(x\right)d_{p}x=\delta_{ij}
\]
is verified directly. Let us show that (\ref{basis_f}) is a~basis
for $L_{norm}^{2}\left(\mathbb{Q}_{p},\:u\left(x\right)d_{p}x\right)$.
Clearly, the countable family of the supports of all $p$-adic spheres
\[
S_{r}=\left\{ x\in\mathbb{Q}_{p}:\:|x|_{p}=p^{r}\right\}
\]
of the form
\[
\delta\left(|x|_{p}-p^{r}\right)=\begin{cases}
1, & x\in S_{r},\\
0, & x\notin S_{r}
\end{cases}
\]
forms a basis for $L_{norm}^{2}\left(\mathbb{Q}_{p},\:u\left(x\right)d_{p}x\right)$.
We have $\delta\left(|x|_{p}-p^{r}\right)=\varOmega\left(|x|_{p}p^{-r}\right)-\varOmega\left(|x|_{p}p^{-r+1}\right)$,
and so it suffices to show that, for any~$r$, the function $\varOmega\left(|x|_{p}p^{-r}\right)$
can be expanded in functions~(\ref{basis_f}). We claim that the
norm of the following function is zero. Indeed,
\[
\left\Vert \varOmega\left(|x|_{p}p^{-r}\right)-V_{r}\sum_{i=r+1}^{\infty}V_{i-1}^{-\tfrac{1}{2}}\left(1-\dfrac{V_{i-1}}{V_{i}}\right)^{\tfrac{1}{2}}\phi_{i}\left(x\right)\right\Vert ^{2}=
\]
\[
=\intop_{\mathbb{Q}_{p}}\varOmega\left(|x|_{p}p^{-r}\right)m\left(x\right)d_{p}x-
\]
\[
-2V_{r}\sum_{i=r+1}^{\infty}V_{i-1}^{-\tfrac{1}{2}}\left(1-\dfrac{V_{i-1}}{V_{i}}\right)^{\tfrac{1}{2}}\intop_{\mathbb{Q}_{p}}\varOmega\left(|x|_{p}p^{-r}\right)\phi_{i}\left(x\right)m\left(x\right)d_{p}x+
\]
\[
+V_{r}^{2}\sum_{i=r+1}^{\infty}V_{i-1}^{-1}\left(1-\dfrac{V_{i-1}}{V_{i}}\right)\intop_{\mathbb{Q}_{p}}\left(\phi_{i}\left(x\right)\right)^{2}m\left(x\right)d_{p}x=
\]
\[
=V_{r}-2V_{r}\sum_{i=r+1}^{\infty}V_{i-1}^{-\tfrac{1}{2}}\left(1-\dfrac{V_{i-1}}{V_{i}}\right)^{\tfrac{1}{2}}V_{i-1}^{-\tfrac{1}{2}}V_{r}\left(1-\dfrac{V_{i-1}}{V_{i}}\right)^{\tfrac{1}{2}}+
\]
\[
+V_{r}^{2}\sum_{i=r+2}^{\infty}V_{i-1}^{-1}\left(1-\dfrac{V_{i-1}}{V_{i}}\right)=
\]
\[
=V_{r}-V_{r}^{2}\left(\sum_{i=r+1}^{\infty}V_{i-1}^{-1}-\sum_{i=r+2}^{\infty}V_{i}^{-1}\right)=0
\]
As a result, we have
\begin{equation}
\varOmega\left(|x|_{p}p^{-r}\right)=V_{r}\sum_{i=r+1}^{\infty}V_{i-1}^{-\tfrac{1}{2}}\left(1-\dfrac{V_{i-1}}{V_{i}}\right)^{\tfrac{1}{2}}\phi_{i}\left(x\right),\label{Omega}
\end{equation}
completing the proof of the theorem.\end{proof}

Any function $f\left(x\right)$ from $L_{norm}^{2}\left(\mathbb{Q}_{p},\:u\left(x\right)d_{p}x\right)$
can be expanded as
\[
f\left(x\right)=\sum_{i=-\infty}^{\infty}f_{i}\phi_{i}\left(x\right),
\]
where
\[
f_{i}=\intop_{\mathbb{Q}_{p}}\phi_{i}\left(x\right)f\left(x\right)m\left(x\right)d_{p}x.
\]
We note that functions (\ref{basis_f}) have the property
\[
\intop_{\mathbb{Q}_{p}}\phi_{i}\left(x\right)m\left(x\right)d_{p}x=0.
\]

The following result holds \cite{BZ_2}.

\begin{theorem} Functions \eqref{basis_f} are eigenfunctions of
the operator that generalizes the Vladimirov operator in $L_{a}^{2}\left(\mathbb{Q}_{p},\:m\left(x\right)d_{p}x\right)$
\begin{equation}
D_{m(x)}^{\alpha}f\left(x\right)=\intop_{\mathbb{Q}_{p}}d_{p}ym\left(y\right)\frac{f\left(y\right)-f\left(x\right)}{\left\vert x-y\right\vert _{p}^{\alpha+1}}\label{Vlad_Q_p_gen}
\end{equation}
with the eigenvalues
\begin{equation}
\lambda_{i}=-\left(1-p^{-\left(\alpha+1\right)}\right)\sum_{j=i}^{\infty}p^{-j\left(\alpha+1\right)}V_{j}..\label{E_V}
\end{equation}
\label{th2}\end{theorem}

\begin{lemma}The following formula holds
\[
\intop_{\mathbb{Q}_{p}}\phi_{i}\left(x\right)\phi_{j}\left(x\right)\phi_{k}\left(x\right)m\left(x\right)d_{p}x=
\]
\[
=\delta_{ij}\delta_{jk}V_{k-1}^{-\tfrac{1}{2}}\left(1-\dfrac{V_{k-1}}{V_{k}}\right)^{-\tfrac{3}{2}}\left(1-3\dfrac{V_{k-1}}{V_{k}}+2\dfrac{V_{k-1}^{2}}{V_{k}^{2}}\right)+
\]
\[
+\delta_{ij}\delta_{i<k}V_{k-1}^{-\tfrac{1}{2}}\left(1-\dfrac{V_{k-1}}{V_{k}}\right)^{\tfrac{1}{2}}+\delta_{ik}\delta_{i<j}V_{j-1}^{-\tfrac{1}{2}}\left(1-\dfrac{V_{j-1}}{V_{j}}\right)^{\tfrac{1}{2}}+
\]
\[
+\delta_{jk}\delta_{j<i}V_{i-1}^{-\tfrac{1}{2}}\left(1-\dfrac{V_{i-1}}{V_{i}}\right)^{\tfrac{1}{2}},
\]
where
\[
\delta_{j<i}=\begin{cases}
1, & j<i,\\
0, & j\geq i.
\end{cases}
\]
\label{lemma2} \end{lemma}

The lemma is is proved by direct calculation.

\begin{theorem} Let $f\left(x\right)$, $g\left(x\right)\in L_{norm}^{2}\left(\mathbb{Q}_{p},\:u\left(x\right)d_{p}x\right)$,
\[
f\left(x\right)=\sum_{i=-\infty}^{\infty}f_{i}\phi_{i}\left(x\right),\;g\left(x\right)=\sum_{i=-\infty}^{\infty}g_{i}\phi_{i}\left(x\right).
\]
Then
\[
f\left(x\right)g\left(x\right)=\sum_{k=-\infty}^{\infty}h_{k}\phi_{k}\left(x\right),
\]
where
\[
h_{k}=f_{k}g_{k}V_{k-1}^{-\tfrac{1}{2}}\left(1-\dfrac{V_{k-1}}{V_{k}}\right)^{-\tfrac{3}{2}}\left(1-3\dfrac{V_{k-1}}{V_{k}}+2\dfrac{V_{k-1}^{2}}{V_{k}^{2}}\right)+
\]
\[
+V_{k-1}^{-\tfrac{1}{2}}\left(1-\dfrac{V_{k-1}}{V_{k}}\right)^{\tfrac{1}{2}}\sum_{i=-\infty}^{k-1}f_{i}g_{i}+f_{k}\sum_{i=k+1}^{\infty}g_{i}V_{i-1}^{-\tfrac{1}{2}}\left(1-\dfrac{V_{i-1}}{V_{i}}\right)^{\tfrac{1}{2}}+
\]
\[
+g_{k}\sum_{i=k+1}^{\infty}f_{i}V_{i-1}^{-\tfrac{1}{2}}\left(1-\dfrac{V_{i-1}}{V_{i}}\right)^{\tfrac{1}{2}}.
\]
\label{th3} \end{theorem}

\begin{proof} We write
\[
f\left(x\right)g\left(x\right)=\sum_{i=-\infty}^{\infty}\sum_{j=-\infty}^{\infty}f_{i}g_{j}\phi_{i}\left(x\right)\phi_{j}\left(x\right)=\sum_{k=-\infty}^{\infty}h_{k}\phi_{k}
\left(x\right),
\]
where
\[
h_{k}=\sum_{i=-\infty}^{\infty}\sum_{j=-\infty}^{\infty}f_{i}g_{j}\intop_{\mathbb{Q}_{p}}\phi_{i}\left(x\right)\phi_{j}\left(x\right)\phi_{k}\left(x\right)m\left(x\right)d_{p}x.
\]
Now the theorem follows from Lemma~\ref{lemma2}.\end{proof}

The following result is a direct corollary to Theorem \ref{th3} taking
into account the expansion formula

\[
\varOmega\left(|x|_{p}\right)=\sum_{i=0}^{\infty}\dfrac{V_{0}}{V_{i}}\phi_{k}\left(x\right).
\]

\begin{corollary} Let $f\left(x\right)\in L_{norm}^{2}\left(\mathbb{Q}_{p},\:u\left(x\right)d_{p}x\right)$
and $f\left(x\right)=\sum_{k=-\infty}^{\infty}f_{k}\phi_{k}\left(x\right)$.
Then the production of functions $\varOmega\left(|x|_{p}\right)f\left(x\right)$
is expanded in the basis \eqref{basis_f} as
\[
\varOmega\left(|x|_{p}\right)f\left(x\right)=\sum_{k=-\infty}^{\infty}C_{k}\phi_{k}\left(x\right),
\]
where
\[
C_{k}=f_{k}V_{0}V_{k-1}^{-1}\left(1-\dfrac{V_{k-1}}{V_{k}}\right)^{-1}\left(1-3\dfrac{V_{k-1}}{V_{k}}+2\dfrac{V_{k-1}^{2}}{V_{k}^{2}}\right)+
\]
\[
+f_{k}\sum_{i=k}^{\infty}V_{0}V_{i-1}^{-1}\left(1-\dfrac{V_{i-1}}{V_{i}}\right)-2f_{k}V_{0}V_{k-1}^{-1}\left(1-\dfrac{V_{k-1}}{V_{k}}\right)+
\]
\begin{equation}
+V_{0}V_{k-1}^{-\tfrac{1}{2}}\left(1-\dfrac{V_{k-1}}{V_{k}}\right)^{\tfrac{1}{2}}\sum_{i=1}^{\infty}f_{i}V_{i-1}^{-\tfrac{1}{2}}\left(1-\dfrac{V_{i-1}}{V_{i}}\right)^{\tfrac{1}{2}}\label{Omega_f_1}
\end{equation}
for $k\geq1$ and
\begin{equation}
C_{k}=f_{k}\sum_{i=k+1}^{\infty}V_{0}V_{i-1}^{-1}\left(1-\dfrac{V_{i-1}}{V_{i}}\right)\label{Omega_f_2}
\end{equation}
for $k<1$.\label{cor1}\end{corollary}

\section*{Appendix~B. Estimate of $\delta_{k}$ for large $k$}

To estimate $\delta_{k}$ for large $k$ it is convenient to write
equation (\ref{delta_k}) as
\begin{equation}
\dfrac{p^{-\left(\alpha-1\right)k}p^{-\alpha+1}}{\delta_{k}}=\sum_{i=1}^{k}\dfrac{p^{-\left(\alpha-1\right)i}}{1-p^{-k-1+i}\left(1+\delta_{k}\right)}-\sum_{i=k+2}^{\infty}\dfrac{p^{-\alpha i}p^{k+1}}{1-p^{-i+k+1}\left(1+\delta_{k}\right)}-\dfrac{1}{p^{\alpha-1}-1}.\label{delta_k_second}
\end{equation}
In what follows we shall require the following result. \begin{lemma}If
$0<x<1$, then
\begin{equation}
\dfrac{1}{1-x}>1+x.\label{eq_1}
\end{equation}
If $0<x<y<1$, then
\begin{equation}
\dfrac{1}{1-x}\leq1+cx,\label{eq_2}
\end{equation}
where $c=\dfrac{1}{1-y}$.\label{lemma3} \end{lemma}

\begin{proof}
\[
\dfrac{1}{1-x}-\left(1+x\right)=\dfrac{x^{2}}{1-x}>0.
\]
\[
\dfrac{1}{1-x}-\left(1+cx\right)=\dfrac{\left(1-c\right)x+cx^{2}}{1-x}\leq x\dfrac{\left(1-c\right)+cy}{1-x}=x\dfrac{1-c\left(1-y\right)}{1-x}=\dfrac{1-1}{1-x}=0.
\]
\end{proof}

Consider the case $\alpha>2$. By (\ref{eq_1}) and (\ref{eq_2})
and Lemma~\ref{lemma3} with $\alpha\neq2$ we estimate the terms
on the right of~(\ref{delta_k_second}) as follows:
\[
\sum_{i=1}^{k}\dfrac{p^{-\left(\alpha-1\right)i}}{1-p^{-k-1+i}\left(1+\delta_{k}\right)}>\sum_{i=1}^{k}p^{-\left(\alpha-1\right)i}\left(1+p^{-k-1+i}\left(1+\delta_{k}\right)\right)=
\]
\[
=\dfrac{1-p^{-\left(\alpha-1\right)k}}{p^{\alpha-1}-1}+\left(1+\delta_{k}\right)\dfrac{1-p^{-\left(\alpha-2\right)k}}{p^{\alpha-2}-1}p^{-k-1},
\]
\[
\sum_{i=1}^{k}\dfrac{p^{-\left(\alpha-1\right)i}}{1-p^{-k-1+i}\left(1+\delta_{k}\right)}\leq\sum_{i=1}^{k}p^{-\left(\alpha-1\right)i}\left(1+\dfrac{1+\delta_{k}}{1-p^{-1}}p^{-k-1+i}\right)=
\]
\[
=\dfrac{1-p^{-\left(\alpha-1\right)k}}{p^{\alpha-1}-1}+\dfrac{1+\delta_{k}}{p-1}\dfrac{1-p^{-\left(\alpha-2\right)k}}{p^{\alpha-2}-1}p^{-k},
\]
Moreover, the following estimate holds:
\[
-\sum_{i=k+2}^{\infty}\dfrac{p^{-\alpha i}p^{k+1}}{1-p^{-i+k+1}\left(1+\delta_{k}\right)}>-\sum_{i=k+2}^{\infty}\dfrac{p^{-\alpha i}p^{k+1}}{1-p^{-1}\left(1+\delta_{k}\right)}=-\dfrac{p^{-\left(\alpha-1\right)k}p^{-2\alpha+1}}{\left(1-p^{-1}\right)\left(1-p^{-\alpha}\right)\left(1+\delta_{k}\right)},
\]
\[
-\sum_{i=k+2}^{\infty}\dfrac{p^{-\alpha i}p^{k+1}}{1-p^{-i+k+1}\left(1+\delta_{k}\right)}<-\sum_{i=k+2}^{\infty}p^{-\alpha i}p^{k+1}=-\dfrac{p^{-\left(\alpha-1\right)k}p^{-2\alpha+1}}{1-p^{-\alpha}}
\]
Applying these estimates, we have
\begin{equation}
\dfrac{p^{-\left(\alpha-1\right)k}p^{-\alpha+1}}{\delta_{k}}>-\dfrac{p^{-\left(\alpha-1\right)\left(k+1\right)}}{1-p^{-\alpha+1}}+\dfrac{\left(1+\delta_{k}\right)\left(1-p^{-\left(\alpha-2\right)k}\right)}{p^{\alpha-2}-1}p^{-k-1}-\dfrac{p^{-\left(\alpha-1\right)k}p^{-2\alpha+1}}{\left(1-p^{-1}\right)\left(1-p^{-\alpha}\right)\left(1+\delta_{k}\right)},\label{neq_1}
\end{equation}
\begin{equation}
\dfrac{p^{-\left(\alpha-1\right)k}p^{-\alpha+1}}{\delta_{k}}<-\dfrac{p^{-\left(\alpha-1\right)k}}{p^{\alpha-1}-1}+\dfrac{1+\delta_{k}}{p-1}\dfrac{1-p^{-\left(\alpha-2\right)k}}{p^{\alpha-2}-1}p^{-k}-\dfrac{p^{-\left(\alpha-1\right)k}p^{-2\alpha+1}}{1-p^{-\alpha}}.\label{neq_2}
\end{equation}
Next we need the following assumption on the behaviour of $\delta_{k}$
as $k\rightarrow\infty$. We make the ansatz that $\delta_{k}\rightarrow0$
as $k\rightarrow\infty$ and solve these inequalities with this ansatz.
Then, in the asymptotic limit $k\rightarrow\infty$, the last two
inequalities may be written as
\begin{equation}
\left(1-p^{-1}\right)\left(1-p^{-\alpha+2}\right)p^{-\left(\alpha-2\right)k}\left(1+o\left(1\right)\right)<\delta_{k}<\left(1-p^{-\alpha+2}\right)p^{-\left(\alpha-2\right)k}\left(1+o\left(1\right)\right),\:\alpha>2.\label{delta_1}
\end{equation}

Consider the case $\alpha=2$. In this case we estimate the terms
on the right of (\ref{delta_k_second}) as
\[
\sum_{i=1}^{k}\dfrac{p^{-\left(\alpha-1\right)i}}{1-p^{-k-1+i}\left(1+\delta_{k}\right)}>\sum_{i=1}^{k}p^{-i}\left(1+p^{-k-1+i}\left(1+\delta_{k}\right)\right)=
\]
\[
=\dfrac{1-p^{-k}}{p^{\alpha-1}-1}+\left(1+\delta_{k}\right)kp^{-k-1},
\]
\[
\sum_{i=1}^{k}\dfrac{p^{-\left(\alpha-1\right)i}}{1-p^{-k-1+i}\left(1+\delta_{k}\right)}\leq\sum_{i=1}^{k}p^{-i}\left(1+\dfrac{1+\delta_{k}}{1-p^{-1}}p^{-k-1+i}\right)=
\]
\[
=\dfrac{1-p^{-k}}{p^{\alpha-1}-1}+\dfrac{1+\delta_{k}}{p-1}kp^{-k},
\]
\[
-\sum_{i=k+2}^{\infty}\dfrac{p^{-\alpha i}p^{k+1}}{1-p^{-i+k+1}\left(1+\delta_{k}\right)}>-\sum_{i=k+2}^{\infty}\dfrac{p^{-2i}p^{k+1}}{1-p^{-1}\left(1+\delta_{k}\right)}=-\dfrac{p^{-k}p^{-3}}{\left(1-p^{-1}\right)\left(1-p^{-2}\right)\left(1+\delta_{k}\right)},
\]
\[
-\sum_{i=k+2}^{\infty}\dfrac{p^{-\alpha i}p^{k+1}}{1-p^{-i+k+1}\left(1+\delta_{k}\right)}<-\sum_{i=k+2}^{\infty}p^{-2i}p^{k+1}=-\dfrac{p^{-k}p^{-3}}{1-p^{-2}}.
\]
Using these estimates, we may write
\[
\dfrac{p^{-k}p^{-1}}{\delta_{k}}>-\dfrac{p^{-k-1}}{1-p^{-1}}+\left(1+\delta_{k}\right)kp^{-k-1}-\dfrac{p^{-k}}{\left(p-1\right)\left(p^{2}-1\right)\left(1+\delta_{k}\right)},
\]
\[
\dfrac{p^{-k}p^{-1}}{\delta_{k}}<-\dfrac{p^{-k}}{p^{\alpha-1}-1}+\dfrac{1+\delta_{k}}{p-1}kp^{-k}-\dfrac{p^{-k}p^{-3}}{1-p^{-2}}.
\]
We again make make the ansatz that $\delta_{k}\rightarrow0$ as $k\rightarrow\infty$.
With this ansatz the last two inequalities read as
\begin{equation}
\left(1-p^{-1}\right)k^{-1}\left(1+o\left(1\right)\right)<\delta_{k}<k^{-1}\left(1+o\left(1\right)\right),\:\alpha>2.\label{delta_2}
\end{equation}

Consider the case $1<\alpha<2$. One may show that inequalities (\ref{neq_1})
and (\ref{neq_2}) with the ansatz $\delta_{k}\rightarrow0$ as $k\rightarrow\infty$
lead to a~contradiction. Hence, in this case we make a~different
ansatz that $\delta_{k}\rightarrow\operatorname{const}$ as $k\rightarrow\infty$.
This gives us the inequalities
\[
\dfrac{p^{-\left(\alpha-1\right)k}p^{-\alpha+1}}{\delta_{k}}>-\dfrac{p^{-\left(\alpha-1\right)\left(k+1\right)}}{1-p^{-\alpha+1}}+\dfrac{\left(1+\delta_{k}\right)\left(1-p^{-\left(\alpha-2\right)k}\right)}{p^{\alpha-2}-1}p^{-k-1}-\dfrac{p^{-\left(\alpha-1\right)k}p^{-2\alpha+1}}{\left(1-p^{-1}\right)\left(1-p^{-\alpha}\right)\left(1+\delta_{k}\right)},
\]
\[
\dfrac{p^{-\left(\alpha-1\right)k}p^{-\alpha+1}}{\delta_{k}}<-\dfrac{p^{-\left(\alpha-1\right)k}}{p^{\alpha-1}-1}+\dfrac{1+\delta_{k}}{1-p^{-1}}\dfrac{1-p^{-\left(\alpha-2\right)k}}{p^{\alpha-2}-1}p^{-k-1}-\dfrac{p^{-\left(\alpha-1\right)k}p^{-2\alpha+1}}{1-p^{-\alpha}},
\]
which may be written as
\[
\delta_{k}^{2}+\left(1-\dfrac{p^{-\alpha+2}-1}{1-p^{-\alpha+1}}-\dfrac{p^{-\alpha+2}-1}{\left(1-p^{-1}\right)\left(p^{\alpha}-1\right)}\right)\delta_{k}-\left(p^{-\alpha+2}-1\right)<0,
\]
\[
\delta_{k}^{2}+\left(1-\dfrac{\left(1-p^{-1}\right)\left(p^{-\alpha+2}-1\right)}{1-p^{-\alpha+1}}-\dfrac{\left(1-p^{-1}\right)\left(p^{-\alpha+2}-1\right)}{p^{\alpha}-1}\right)\delta_{k}-\left(1-p^{-1}\right)\left(p^{-\alpha+2}-1\right)>0.
\]
We denote by $\delta_{\max}$ the maximal positive root of the equation
\[
\delta_{k}^{2}+\left(1-\dfrac{p^{-\alpha+2}-1}{1-p^{-\alpha+1}}-\dfrac{p^{-\alpha+2}-1}{\left(1-p^{-1}\right)\left(p^{\alpha}-1\right)}\right)\delta_{k}-\left(p^{-\alpha+2}-1\right)=0
\]
and denote by $\delta_{\min}$ the maximal positive root of the equation
\[
\delta_{k}^{2}+\left(1-\dfrac{\left(1-p^{-1}\right)\left(p^{-\alpha+2}-1\right)}{1-p^{-\alpha+1}}-\dfrac{\left(1-p^{-1}\right)\left(p^{-\alpha+2}-1\right)}{p^{\alpha}-1}\right)\delta_{k}-\left(1-p^{-1}\right)\left(p^{-\alpha+2}-1\right)=0.
\]
Then
\[
\delta_{\min}\left(1+o\left(1\right)\right)<\delta_{k}<\delta_{\max}\left(1+o\left(1\right)\right).
\]
One may show that $\lim_{\alpha\rightarrow2}\delta_{\max}=0$, $\lim_{\alpha\rightarrow1}\delta_{\max}=\infty$
and $\lim_{\alpha\rightarrow2}\delta_{\min}=0$, $\lim_{\alpha\rightarrow1}\delta_{\min}=\infty$.

Combining these results, we write down the lower and upper estimates
for $\delta_{k}$ as $k\rightarrow\infty$ for different~$\alpha$
as
\begin{equation}
D_{\min}<\delta_{k}<D_{\max},\label{delta_gen}
\end{equation}
where
\[
D_{\min}=\left(1+o\left(1\right)\right)\begin{cases}
\left(1-p^{-1}\right)\left(1-p^{-\alpha+2}\right)p^{-\left(\alpha-2\right)k}, & \alpha>2,\\
\left(1-p^{-1}\right)k^{-1}, & \alpha=2,\\
\delta_{\min}, & 1<\alpha<2,
\end{cases}
\]
\[
D_{\max}=\left(1+o\left(1\right)\right)\begin{cases}
\left(1-p^{-\alpha+2}\right)p^{-\left(\alpha-2\right)k}, & \alpha>2,\\
k^{-1}, & \alpha=2,\\
\delta_{\max}, & 1<\alpha<2.
\end{cases}
\]

\section*{Appendix~C. Estimate of $b_{k}$ for large $k$}

Let us estimate the residues $b_{k}$ as $k\rightarrow\infty$. To
this aim we write \eqref{b_k} as
\begin{equation}
b_{k}=\dfrac{p^{\alpha}-1}{\left(p^{\alpha-1}-1\right)^{2}}\left[\sum_{i=1}^{k}\dfrac{p^{-\alpha i}}{\left(p^{-k-1}+p^{-k-1}\delta_{k}-p^{-i}\right)^{2}}+\sum_{i=k+2}^{\infty}\dfrac{p^{-\alpha i}}{\left(p^{-k-1}+p^{-k-1}\delta_{k}-p^{-i}\right)^{2}}+\dfrac{p^{2-\alpha}p^{-\left(\alpha-2\right)k}}{\delta_{k}^{2}}\right]^{-1}.\label{b_k_final}
\end{equation}
Taking into account that in the case $\alpha\geq2$ we have $\delta_{k}\rightarrow0$
as $k\rightarrow\infty$, we write
\[
b_{k}=\dfrac{p^{\alpha}-1}{\left(p^{\alpha-1}-1\right)^{2}}\left[\sum_{i=1}^{k}\dfrac{p^{-\alpha i}}{\left(p^{-k-1}-p^{-i}\right)^{2}}+\sum_{i=k+2}^{\infty}\dfrac{p^{-\alpha i}}{\left(p^{-k-1}-p^{-i}\right)^{2}}+\dfrac{p^{2-\alpha}p^{-\left(\alpha-2\right)k}}{\delta_{k}^{2}}\right]^{-1}\left(1+o\left(1\right)\right).
\]

Let us estimate $b_{k}$ from the above. First, we consider the case
$\alpha>2$. Using \eqref{delta_gen}, we find that
\[
b_{k}<\dfrac{p^{\alpha}-1}{\left(p^{\alpha-1}-1\right)^{2}}\left[\sum_{i=1}^{k}\dfrac{p^{-\alpha i}}{\left(p^{-k-1}-p^{-1}\right)^{2}}+\sum_{i=k+2}^{\infty}\dfrac{p^{-\alpha i}}{\left(p^{-k-1}\right)^{2}}+\dfrac{p^{2-\alpha}p^{\left(\alpha-2\right)k}}{\left(1-p^{-\alpha+2}\right)^{2}}\right]^{-1}\left(1+o\left(1\right)\right).
\]
Summing and neglecting the terms that vanish for $k\rightarrow\infty$,
this gives
\[
b_{k}<\dfrac{\left(p^{\alpha}-1\right)\left(1-p^{-\left(\alpha-2\right)}\right)^{2}}{p^{2-\alpha}\left(p^{\alpha-1}-1\right)^{2}}p^{-\left(\alpha-2\right)k}\left(1+o\left(1\right)\right).
\]
For $\alpha=2$ we similarly have
\[
b_{k}<\dfrac{\left(1-p^{-2}\right)\left(1-p^{-1}\right)^{2}}{\left(p-1\right)^{2}}k^{-2}\left(1+o\left(1\right)\right).
\]
In the case $1<\alpha<2$ we have
\[
b_{k}<\dfrac{p^{\alpha}-1}{\left(p^{\alpha-1}-1\right)^{2}}\left[\dfrac{p^{-\alpha}-p^{-\alpha\left(k+1\right)}}{\left(1-p^{-\alpha}\right)\left(p^{-k-1}-p^{-1}\right)^{2}}+\dfrac{p^{-\alpha\left(k+2\right)}}{\left(1-p^{-\alpha}\right)\left(p^{-k-1}\right)^{2}}+\dfrac{p^{2-\alpha}p^{-\left(\alpha-2\right)k}}{\delta_{\min}^{2}}\right]^{-1}\left(1+o\left(1\right)\right),
\]
which gives
\[
b_{k}<\dfrac{\left(p^{\alpha}-1\right)\delta_{\min}^{2}}{p^{2-\alpha}\left(p^{\alpha-1}-1\right)^{2}}p^{-\left(2-\alpha\right)k}\left(1+o\left(1\right)\right).
\]

Let us now estimate $b_{k}$ from below. For $\alpha>2$, we have
\[
b_{k}>\dfrac{p^{\alpha}-1}{\left(p^{\alpha-1}-1\right)^{2}}\left[\sum_{i=1}^{k}\dfrac{p^{-\alpha i}}{\left(p^{-k-1}-p^{-i}\right)^{2}}+\sum_{i=k+2}^{\infty}\dfrac{p^{-\alpha i}}{\left(p^{-k-1}-p^{-i}\right)^{2}}+\dfrac{p^{2-\alpha}p^{\left(\alpha-2\right)k}}{\left(1-p^{-1}\right)^{2}\left(1-p^{-\alpha+2}\right)^{2}}\right]^{-1}\left(1+o\left(1\right)\right),
\]
which gives
\[
\dfrac{\left(p^{\alpha}-1\right)\left(1-p^{-1}\right)^{2}\left(1-p^{-\left(\alpha-2\right)}\right)^{2}}{p^{2-\alpha}\left(p^{\alpha-1}-1\right)^{2}}p^{-\left(\alpha-2\right)k}\left(1+o\left(1\right)\right)<b_{k}.
\]
In the case $\alpha=2$, we have
\[
b_{k}>\dfrac{\left(p^{2}-1\right)}{\left(p-1\right)^{2}}k^{-2}\left(1+o\left(1\right)\right).
\]
For $1<\alpha<2$, we get
\[
b_{k}>\dfrac{p^{\alpha}-1}{\left(p^{\alpha-1}-1\right)^{2}}\left[\dfrac{p^{-\alpha1}-p^{-\alpha\left(k+1\right)}}{\left(1-p^{-\alpha}\right)\left(p^{-k-1}-p^{-k}\right)^{2}}+\dfrac{p^{-\alpha k}p^{-2\alpha}p^{-\alpha i}}{\left(1-p^{-\alpha}\right)\left(p^{-k-1}-p^{-k-2}\right)^{2}}+\dfrac{p^{2-\alpha}p^{-\left(\alpha-2\right)k}}{\delta_{\max}^{2}}\right]^{-1}\left(1+o\left(1\right)\right),
\]
and hence,
\[
b_{k}>\dfrac{\left(p^{\alpha}-1\right)\delta_{\max}^{2}}{p^{2-\alpha}\left(p^{\alpha-1}-1\right)^{2}}p^{-\left(2-\alpha\right)k}\left(1+o\left(1\right)\right).
\]

Collecting the above inequality, we write the lower and upper estimates
for $b_{k}$ as $k\rightarrow\infty$ for different~$\alpha$ as
follows:
\begin{equation}
a\left(\alpha,k\right)\left(1+o\left(1\right)\right)<b_{k}<b\left(\alpha,k\right)\left(1+o\left(1\right)\right),\label{b_gen}
\end{equation}
where
\[
a\left(\alpha,k\right)=\left(1+o\left(1\right)\right)\begin{cases}
\dfrac{\left(p^{\alpha}-1\right)\left(1-p^{-1}\right)^{2}\left(1-p^{-\alpha+2}\right)^{2}}{p^{2-\alpha}\left(p^{\alpha-1}-1\right)^{2}}p^{-\left(\alpha-2\right)k}, & \alpha>2,\\
\dfrac{\left(p^{2}-1\right)}{\left(p-1\right)^{2}}k^{-2}, & \alpha=2,\\
\dfrac{\left(p^{\alpha}-1\right)\delta_{\max}^{2}}{p^{2-\alpha}\left(p^{\alpha-1}-1\right)^{2}}p^{-\left(2-\alpha\right)k}, & 1<\alpha<2,
\end{cases}
\]
\[
b\left(\alpha,k\right)=\left(1+o\left(1\right)\right)\begin{cases}
\dfrac{\left(p^{\alpha}-1\right)\left(1-p^{-\alpha+2}\right)^{2}}{p^{2-\alpha}\left(p^{\alpha-1}-1\right)^{2}}p^{-\left(\alpha-2\right)k}, & \alpha>2,\\
\dfrac{\left(1-p^{-2}\right)\left(1-p^{-1}\right)^{2}}{\left(p-1\right)^{2}}k^{-2}, & \alpha=2,\\
\dfrac{\left(p^{\alpha}-1\right)\delta_{\min}^{2}}{p^{2-\alpha}\left(p^{\alpha-1}-1\right)^{2}}p^{-\left(2-\alpha\right)k}, & 1<\alpha<2.
\end{cases}
\]

\section*{Appendix~D. Asymptotic estimate of the series}

Here, we shall obtain the estimate for the series
\begin{equation}
S_{1}(t)=\sum\limits _{i=1}^{\infty}i^{-k}a^{-i}\exp\left(-b^{-i}t\right),\quad t\ge0,\;a>1,\;b>1,\;k\geq0\label{ser}
\end{equation}
for large~$t$.

\begin{lemma}For series \eqref{ser} we have the following asymptotic
estimate as $t\rightarrow\infty:$
\[
a^{-1}\left(\ln b\right)^{k-1}t^{-\tfrac{\ln a}{\ln b}}\left(\ln t\right)^{-k}\Gamma\left(\frac{\ln a}{\ln b}\right)\left(1+o\left(1\right)\right)\le S_{1}(t)\le a\left(\ln b\right)^{k-1}t^{-\tfrac{\ln a}{\ln b}}\left(\ln t\right)^{-k}\Gamma\left(\frac{\ln a}{\ln b}\right)\left(1+o\left(1\right)\right),
\]
where $\Gamma\left(z\right)$ is the Gamma function and $o\left(1\right)$
is an infinitesimal quantity as $t\rightarrow\infty$ .\label{lemma4}
\end{lemma}

\begin{proof} Instead of series \eqref{ser}, we consider the series
\[
S_{2}(t)=\sum\limits _{i=2}^{\infty}i^{-k}a^{-i}\exp\left(-b^{-i}t\right)=S_{1}(t)-a^{-1}\exp\left(-b^{-1}t\right).
\]
We note that $\dfrac{1}{x^{k}}a^{-x}$ is a decreasing function and
$\exp\left(-b^{-x}t\right)$ is an increasing function of~$x$. Hence,
on the interval $i\le x\le i+1$ we have
\begin{equation}
\frac{1}{x^{k}}a^{-x}\exp\left(-b^{-(x-1)}t\right)\le a^{-i}\exp\left(-b^{-i}t\right)\le\frac{1}{(x-1)^{k}}a^{-(x-1)}\exp\left(-b^{-x}t\right).\label{A2-1}
\end{equation}
Integrating (\ref{A2-1}) in $x$ from $i$ to $i+1$, we find that
\begin{equation}
a^{-1}\int_{i}^{i+1}\frac{1}{x^{k}}a^{-(x-1)}\exp\left(-b^{-(x-1)}t\right)dx\le a^{-i}\exp\left(-b^{-i}t\right)\le a\int_{i}^{i+1}\frac{1}{(x-1)^{k}}a^{-x}\exp\left(-b^{-x}t\right)dx.\label{A3-1}
\end{equation}
Next, summing (\ref{A3-1}) in~$i$ from~$2$ to~$\infty$, this
gives
\[
S_{\min}(t)\le S_{2}(t)\le S_{\max}(t),
\]
where
\begin{equation}
S_{\min}(t)=a^{-1}\intop_{2}^{\infty}\frac{1}{x^{k}}a^{-(x-1)}\exp\left(-b^{-(x-1)}t\right)dx,\label{S_min}
\end{equation}
\begin{equation}
S_{\max}(t)=a\intop_{2}^{\infty}\frac{1}{(x-1)^{k}}a^{-x}\exp\left(-b^{-x}t\right)dx.\label{S_max}
\end{equation}
In (\ref{S_min}) and (\ref{S_max}) we change the variable of integration
as $b^{-(x-1)}t=y$ and $b^{-x}t=y$, respectively. We have
\[
S_{\min}(t)=a^{-1}\left(\ln b\right)^{-1}t^{-\tfrac{\ln a}{\ln b}}\intop_{0}^{b^{-1}t}\left(\dfrac{\ln t-\ln y}{\ln b}+1\right)^{-k}y^{\tfrac{\ln a}{\ln b}-1}\exp\left(-y\right)dy,
\]
\[
S_{\max}(t)=a\left(\ln b\right)^{-1}t^{-\tfrac{\ln a}{\ln b}}\intop_{0}^{b^{-2}t}\left(\dfrac{\ln t-\ln y}{\ln b}-1\right)^{-k}y^{\tfrac{\ln a}{\ln b}-1}\exp\left(-y\right)dy.
\]
If $k>0$, then, as $t\rightarrow\infty$,
\[
S_{\min}(t)=a^{-1}\left(\ln b\right)^{k-1}t^{-\tfrac{\ln a}{\ln b}}\left(\ln t\right)^{-k}\Gamma\left(\frac{\ln a}{\ln b}\right)\left(1+o\left(1\right)\right),
\]
\[
S_{\max}(t)=a\left(\ln b\right)^{k-1}t^{-\tfrac{\ln a}{\ln b}}\left(\ln t\right)^{-k}\Gamma\left(\frac{\ln a}{\ln b}\right)\left(1+o\left(1\right)\right),
\]
which gives the assertion of the lemma.\end{proof}

%% The Appendices part is started with the command \appendix;
%% appendix sections are then done as normal sections
%% \appendix

%% \section{}
%% \label{}

%% References
%%
%% Following citation commands can be used in the body text:
%% Usage of \cite is as follows:
%%   \cite{key}         ==>>  [#]
%%   \cite[chap. 2]{key} ==>> [#, chap. 2]
%%

%% References with BibTeX database:

\bibliographystyle{elsarticle-num}
%\bibliography{<your-bib-database>}

%% Authors are advised to use a BibTeX database file for their reference list.
%% The provided style file elsarticle-num.bst formats references in the required Procedia style

%% For references without a BibTeX database:

\end{document}